\documentclass
[aps,preprint,preprintnumbers,amsmath,amssymb,nofootinbib]{revtex4}%
\usepackage{amssymb}
\usepackage{amsmath}
\usepackage{bm}
\usepackage{amsthm}
\usepackage{graphicx}
\usepackage{amsfonts}%
\providecommand{\U}[1]{\protect\rule{.1in}{.1in}}
\theoremstyle{plain}

\theoremstyle{definition}
\newtheorem{defn}{Definition}

\theoremstyle{proposition}
\newtheorem{prop}{Proposition}
\theoremstyle{lemma}

\theoremstyle{corollary}

\begin{document}
\title{A note on the problem of proper time in Weyl space-time }

\begin{abstract}
We discuss the question of whether or not a general Weyl structure is a
suitable mathematical model of space-time. This is an issue that has been in
debate since Weyl formulated his unified field theory for the first time. We
do not present the discussion from the point of view of a particular
unification theory, but instead from a more general standpoint, in which the
viability of such a structure as a model of space-time is investigated. Our
starting point is the well known axiomatic approach to space-time given by
Elhers, Pirani and Schild (EPS). In this framework, we carry out an exhaustive
analysis of what is required for a consistent definition for proper time and
show that such a definition leads to the prediction of the so-called ``second
clock effect". We take the view that if, based on experience, we were to
reject space-time models predicting this effect, this could be incorporated as
the last axiom in the EPS approach. Finally, we provide a proof that, in this
case, we are led to a Weyl integrable space-time (WIST) as the most general
structure that would be suitable to model space-time.

\end{abstract}
\author{R. Avalos, F. Dahia, C. Romero}
\affiliation{Departamento de F\'{\i}sica, Universidade Federal da Para\'{\i}ba, Caixa
Postal 5008, 58059-970 Jo\~{a}o Pessoa, PB, Brazil.}
\affiliation{E-mail: rodrigo.avalos@fisica.ufpb.br; fdahia@fisica.ufpb.br; cromero@fisica.ufpb.br}
\maketitle

\section{Introduction}

The axiomatic approach to space-time proposed by Ehlers, Pirani and Schild in
\cite{EPS}, tries to build a suitable mathematical model of space-time from
basic assumptions about the behavior of freely falling particles and the
propagation of light rays. These are considered to be the \textit{primitive
concepts} and other constructions should be derived from these elements plus
some set of hypotheses about their behavior which should be as natural as
possible. This program led EPS to show that the propagation of light
determines a \textit{conformal structure} on space-time, while the motion of
freely falling particles determines a \textit{projective structure} on
space-time. They require these two structures to satisfy a
\textit{compatibility condition}, which then leads to a \textit{Weyl
structure} as the resulting mathematical model. Up to this point, they
conclude that the propagation of light and the motion of freely falling
particles determine a Weyl structure as a natural model of space-time. They
proceed to propose two (apparently different) possible axioms which then
reduce the former to the structure of a \textit{Weyl integrable}
space-time\footnote{Actually, after showing that Weyl's curvature $F$ has to
vanish, they conclude that space-time geometry has a Riemannian
representation. This is true, since the vanishing of $F$ implies (in a simply
connected domain) that the Weyl structure is integrable, and in any Weyl
integrable structure there is a Riemannian representative of the class.}. The
previous axioms express quite general ideas which deal with simple aspects of
the motion of freely falling particles and light rays, while the last axiom
seems to deal with much more delicate details regarding the behavior of
\textit{clocks}. In particular, in order to discuss the behavior of clocks in
this context, we should first have a well-defined and physically sensible
notion of \textit{proper time}. This notion seems to be at the core of the
different discussions, which took place after Weyl's proposal of his unified
theory, concerning the viability of Weyl's space-time as an acceptable model
for physics. Some of these discussions, which involved Weyl, Einstein,
Eddington and Pauli, among others, can be reviewed in \cite{Goenner}. For
instance, it is well-known that the original discussions were based upon the
fact that, in a Weyl structure, parallel transported vectors change their
``norms" according to the following expression:
\begin{equation}
g(V(t),V(t))=g(V_{0},V_{0})e^{\int_{t_{0}}^{t}\omega(\gamma^{\prime}(s))ds}%
\end{equation}
where $\gamma$ is the curve used to make the transport of the initial vector
$V_{0}$, and $\omega$ is Weyl's 1-form field. From the previous expression it
is clear that the norm of such parallel transported vector at a particular
point depends on the curve chosen. It was then argued that this effect was
related with a dependence of the ticking rate of a clock on its history,
concluding that a Weyl structure leads to a \textit{second clock effect}. In
our opinion, this discussion, even though physically well-guided, is not
rigorous, since the association of the tick-tack of a clock with the norm a
parallel transported vector is not trivial, and, even more, norms or lengths
are not well-defined concepts in the context of Weyl's space-time. Thus, we
believe that a detailed and rigorous proof of this intuitive statement is
needed, and for this, a plausible notion for proper time in the context of
Weyl's theory seems to be essential. Although in their paper EPS give a
mathematically well-defined notion of proper time, this notion, in the way
they introduce it, does not seem to be motivated by the motion of freely
falling particles or light rays \cite{EPS}. On the other hand, the existence
of a well-defined and physically sensible definition of proper time in a Weyl
structure has been carefully discussed by V. Perlick \cite{Perlick1}. At first
sight, Perlick's definition might not seem to be very much related with the
one given by EPS. In this paper, we will show that the concept of Perlick's
proper time may be motivated by the axiomatic approach given by EPS, and that
his definition is not only mathematically well-defined, but also physically
sensible. We then will show the equivalence of EPS's and Perlick's definition,
and that the latter leads to the second clock effect. Finally, we will show
that, if we were to rule out models of space-time exhibiting such an effect,
then we would arrive at what is called a ``Weyl integrable space-time"\textit{
}structure. In fact, this claim has been stated previously, using, however, a
different line of reasoning \cite{F-I}.

As related to the EPS paper, we would like to refer the reader to some
interesting results obtained in \cite{Perlick1},\cite{referee1}%
,\cite{referee2},\cite{referee3},\cite{AGS},\cite{AL} and \cite{Arg}, where in particular
\cite{Perlick1} and \cite{referee2} give clear insights on the properties of
Weyl's space-time, providing us with, not only a sensible definition for
proper time, but also suggesting methods to construct standard clocks
(\cite{Perlick1}) and to detect a non-zero \textit{length curvature} $F$
(\cite{referee2}). It is also interesting to stress that in \cite{AGS} and \cite{AL} the authors analyse how, by introducing some rudiments of quantum mechanics and considering the behaviour of \textit{matter waves} in the \textit{classical limit}, the EPS approach can be either modified or supplemented so as to give a Lorentzian structure as the resulting mathematical model of space-time.

\section{Weyl Structures}

Since in this paper we will accept, in view of \cite{EPS}, that the motion of
freely falling particles and light rays determine a Weyl structure on
space-time, we will shortly review the basic definitions of such a structure.

A Weyl manifold is a triple $(M,g,\omega)$, where $M$ is a differentiable
manifold, $g$ a semi-Riemannian metric on $M$, and $\omega$ a 1-form field on
$M$. We assume that $M$ is endowed with a torsion-free linear connection
$\nabla$ satisfying the following \textit{compatibility condition}:
\begin{align}\label{compatibility}
\nabla g=g\otimes\omega.
\end{align}
where, adopting the notation of \cite{Spivak} or \cite{Oneill}, $\nabla
g$ denotes the $(0,3)$-tensor field defined by $\nabla
g(Y,Z,X)\doteq(\nabla_{X}g)(Y,Z)$ , where  $X,Y$ and $Z$ are vector fields.

Results concerning the existence and uniqueness of such a connection are
straightforward and their proofs are analogous to those known for the
Riemannian case \cite{Folland}. It can easily be shown that, in local
coordinates, the components of the Weyl connection $\nabla$ are given by:
\[
\Gamma_{ac}^{u}=\frac{1}{2}g^{bu}(\partial_{a}g_{bc}+\partial_{c}%
g_{ab}-\partial_{b}g_{ca})+\frac{1}{2}g^{bu}(\omega_{b}g_{ca}-\omega_{a}%
g_{bc}-\omega_{c}g_{ab}).
\]

A useful expression concerning Weyl connections is the following
\begin{equation}
\label{compatibility2}X(g(Y,Z))=g(\nabla_{X}Y,Z)+g(Y,\nabla_{X}Z)+\omega
(X)g(Y,Z),
\end{equation}
which is just a restatement of (\ref{compatibility}).

Let us now note that the compatibility condition (\ref{compatibility}) is
invariant under the following group of transformations:%

\begin{equation}
\left\{
\begin{array}
[c]{ll}%
\overline{g}=e^{-f}g & \\
\overline{\omega}=\omega-df &
\end{array}
\right.  \label{weyl group}%
\end{equation}
where $f$ is an arbitrary smooth function defined on $M$. By this we mean that
if $\nabla$ is compatible with $(M,g,\omega),$ then it is also compatible with
$(M,\overline{g},\overline{\omega})$. It is easy to check that these
transformations define an equivalence relation between Weyl manifolds. In this
way, every member of the class is compatible with the same connection, hence
has the same geodesics, curvature tensor and any other property that depends
only on the connection. This is the reason why it is regarded more natural,
when dealing with Weyl manifolds, to consider the whole class of equivalence
$(M,[(g,\omega)])$ rather than working with a particular element of this
class. In this sense, it is argued that only geometrical quantities that are
invariant under (\ref{weyl group}) are of real significance in the case of
Weyl geometry. Following the same line of argument, it is also assumed that
only physical theories and physical quantities presenting this kind of
invariance should be considered of interest in this context. We will adopt
this point of view which, in fact, is not discussed in \cite{EPS}, where Weyl
transformations do not play an important role, while they are a fundamental ingredient in Weyl's 
original approach\footnote{See, for instance, the footnote in page 68 of
\cite{EPS}}. To conclude this section, we remark that when the 1-form field
$\omega$ is an exact form, then the Weyl structure is called
\textit{integrable}. In view of our previous discussion, in contrast to what
is done in \cite{EPS}, we will make a distintion between an integrable Weyl
structure and the Riemannian element in the class.

\section{Proper Time}

In this section, the main idea is to introduce a definition of \textit{proper
time} which would naturally fit the axiomatic approach proposed by Ehlers,
Pirani and Schild. In this work, we will disregard the last of the
EPS\footnote{In \cite{EPS}, in order to reduce the Weyl structure to a
Riemannian one, two possible additional axioms, regarding the behaviour of
clocks, are introduced. The first one is stated as follows. Given two freely
falling, infinitesimally proximate clocks $C_{1}$ and $C_{2}$, if we consider
a regular sequence of events $(p_{1},p_{2},\dots)$ in the world line of
$C_{1}$ determined by the ticking of this clock, and the Einstein-simultaneous
sequence of events $(q_{1},q_{2},\dots)$ in the world line of $C_{2}$, then
$(q_{1},q_{2},\dots)$ should also be a regular sequence of events. With the
help of the geodesic deviation equation, EPS show that this hypothesis reduces
the Weyl structure to a Riemannian one. Another way to achieve the same the
same goal, is to consider as an axiom that the ``norm" of parallel transported
vector fields at a point can not depend on the curve chosen to make the
transport, associating this ``norm", for the case of time-like curves, to the
ticking rate of clocks.} axioms, however retaining that the motion of freely
falling particles and light rays determine a Weyl structure as a suitable
mathematical model of space-time. In this way we will think of the space-time
geometry as given by a structure of the form $(M,[(g,\omega)])$. In this
scenario, an axiomatic definition for proper time, guided by the EPS
framework, should come as a natural way of distinguishing the proper time
parametrization among all the possible parametrizations of a time-like curve. A
natural way of doing this for the case of free falling particles (pregeodesics
of the Weyl structure), using the primitive concepts available in our
axiomatic framework, is to distinguish the parametrization which makes the
free-falling particle satisfy the geodesic equation, which will be clearly
defined uniquely up to an affine transformation. This idea has already been
suggested in \cite{Salim-Poulis}. A particle $\gamma:I\mapsto M$, $u\mapsto\gamma(u)$,
is freely falling if it is a pregeodesic, that is, if
\begin{align}\label{pt1}
\frac{D\gamma^{\prime}(u)}{du}=f(\gamma(u))\gamma^{\prime}(u).
\end{align}
Our principle is that proper time parametrization is a reparametrization of
$\gamma$ that transforms it into a geodesic. Then, if $\gamma(\tau)$ is a
time-like curve representing a freely falling particle, it is straightforward
to see that $\gamma$ is parametrized by proper time if, and only if
\begin{align}\label{pt2}
g(\gamma^{\prime}(\tau),\frac{D\gamma^{\prime}(\tau)}{d\tau})=0.
\end{align}

This notion of proper time fits naturally into the picture proposed by EPS,
since we are using the basic concepts of their construction to motivate it.
Our next step should be to generalize this definition for arbitrary time-like
curves, which represent arbitrary particles. Obviously, since, in general,
particles are not freely falling, (\ref{pt1}) does not generalize naturally to
the general case. However, (\ref{pt2}) does, and, as we have just seen,
(\ref{pt1}) and (\ref{pt2}) are equivalent in the case of freely falling
particles. This leads us to the following definition, which is precisely the
one given by V. Perlick in \cite{Perlick1}:

\begin{defn}
A time-like curve $\gamma:I\mapsto M$, $u\mapsto\gamma(u)$, is called a
\textit{standard clock} if $\frac{D\gamma^{\prime}}{du}$ is orthogonal to
$\gamma^{\prime}(u)$.
\end{defn}

In order to check that this definition is mathematically consistent with the
fact that we are working with a Weyl structure, we should check that it is
independent of the representative member of the class chosen to carry out the
computations. Suppose that for some particular $g\in\lbrack g]$ we have
$g(\gamma^{\prime},\frac{D\gamma^{\prime}}{du})=0$. Then, since $\nabla$ is
independent of the choice of the representative, so is the covariant
derivative, and hence $\frac{D\gamma^{\prime}}{du}$ does not depend on this
choice. Also, since any other $\tilde{g}\in\lbrack g]$ is related to $g$ by a
conformal transformation, then orthogonality of vectors is preserved. Thus,
the definition is consistent in the whole class. Also, in order for this
definition to be sensible, we should show that any time-like curve can be
parametrized by this kind of parametrization. To see this, consider a
time-like curve
\begin{align*}
\gamma:I  &  \mapsto M\\
t  &  \mapsto\gamma(t)
\end{align*}
We are looking for a reparametrization of $\gamma$ such that the
reparametrized curve is a standard clock. This means that we are looking for a
diffeomorphism
\begin{align*}
\mu:I  &  \mapsto I^{\prime}\\
t  &  \mapsto\tau
\end{align*}
such that
\begin{align*}
\tilde{\gamma}=\gamma\circ\mu^{-1}:I^{\prime}  &  \mapsto M\\
\tau &  \mapsto\gamma(\mu^{-1}(\tau))
\end{align*}
is a standard clock. With this set-up, we see that $\gamma=\tilde{\gamma}%
\circ\mu:I\mapsto M$ and
\begin{align}
\label{pt3}\frac{D\gamma^{\prime}(t)}{dt}=\frac{d^{2}\mu}{dt^{2}}\tilde
{\gamma}^{\prime}(\mu(t))+(\frac{d\mu}{dt})^{2}\frac{D\tilde{\gamma}^{\prime}%
}{d\tau}(\mu(t)).
\end{align}
Also from (\ref{pt3}) we get
\[%
\begin{split}
g(\frac{D\gamma^{\prime}(t)}{dt},\gamma^{\prime}(t))  &  =\frac{d^{2}\mu
}{dt^{2}}g(\tilde{\gamma}^{\prime}(\mu(t)),\gamma^{\prime}(t))+(\frac{d\mu
}{dt})^{2}g(\frac{D\tilde{\gamma}^{\prime}}{d\tau}(\mu(t)),\gamma^{\prime
}(t))\\
&  =\frac{1}{\frac{d\mu}{dt}}\frac{d^{2}\mu}{dt^{2}}g(\gamma^{\prime
}(t),\gamma^{\prime}(t))+(\frac{d\mu}{dt})^{3}g(\frac{D\tilde{\gamma}^{\prime
}}{d\tau}(\mu(t)),\tilde{\gamma}^{\prime}(\mu(t)))
\end{split}
\]
Then, the following equation is satisfied:
\[
\frac{d^{2}\mu}{dt^{2}}-\frac{g(\gamma^{\prime}(t),\frac{D\gamma^{\prime}%
(t)}{dt})}{g(\gamma^{\prime}(t),\gamma^{\prime}(t))}\frac{d\mu}{dt}%
+(\frac{d\mu}{dt})^{4}\frac{g(\frac{D\tilde{\gamma}^{\prime}}{d\tau}%
(\mu(t)),\tilde{\gamma}^{\prime}(\mu(t)))}{g(\gamma^{\prime}(t),\gamma
^{\prime}(t))}=0.
\]
From this last equation we see that $\tilde{\gamma}$ is a standard clock if,
and only if, the reparametrization $\mu$ satisfies the following differential
equation:
\begin{align}
\label{pt4}\frac{d^{2}\mu}{dt^{2}}-\frac{g(\gamma^{\prime}(t),\frac
{D\gamma^{\prime}(t)}{dt})}{g(\gamma^{\prime}(t),\gamma^{\prime}(t))}%
\frac{d\mu}{dt}=0
\end{align}

We then see that if a reparametrization $\mu^{-1}$ makes $\gamma\circ\mu^{-1}$
a standard clock, then it satisfies (\ref{pt4}). Conversely, given a solution
of (\ref{pt4}) with initial conditions such that $\frac{d\mu}{dt}(t_{0})\neq
0$, then there is a neighborhood $I$ of $t_{0}$ where $\mu:I\mapsto
\mu(I)=I^{\prime}$ is a diffeomorphism, and hence the reparametrization
$\gamma\circ\mu^{-1}:I^{\prime}\mapsto M$ will be a standard clock. Thus, we
can state that given a time-like curve $\gamma(t)$, there is a
reparametrization which makes it a standard clock if, and only if, the
equation (\ref{pt4}) admits a solution with $\frac{d\mu}{dt}(t_{0})\neq0$.
Since this type of differential equation always admits solutions for given
initial data $\mu(t_{0}),\mu^{\prime}(t_{0})$, then a given time-like curve
$\gamma$ can always be reparametrized in a neighborhood of any point, so as to
make it a standard clock. Seeing that this definition makes mathematical
sense, we will use it to define \textit{proper time}.

\begin{defn}
We will say that a time-like curve $\gamma$ is parametrized by proper time if
the parametrized curve is a standard clock.
\end{defn}

Concerning the above definitions, the following remarks should be
made:\newline1) Using the same notations as above, the map $\mu:I\mapsto
I^{\prime}$, maps an arbitrary parametrization $t$ to proper time $\tau
=\mu(t)$.\newline2) Since (\ref{pt4}) is linear, if $\mu$ is a solution then
$\tilde{\mu}(t)=a\mu(t)+b$ is also a solution, where $a,b$ are arbitrary
constants. These constants are determined by the initial conditions, and they
just represent the scale $(a)$ and the zero of the clock $(b)$.

\bigskip It will be important to note that a general solution of equation
(\ref{pt4}) can be obtained. In order to do this, first note that from
(\ref{compatibility2}) we have
\[
\frac{d}{dt}\ln(-g(\gamma^{\prime}(t)),\gamma^{\prime}(t))=\frac{1}%
{g(\gamma^{\prime}(t),\gamma^{\prime}(t))}\big\{2g(\gamma^{\prime}%
(t),\frac{D\gamma^{\prime}(t)}{dt})+\omega(\gamma^{\prime}(t))g(\gamma
^{\prime}(t),\gamma^{\prime}(t))\big\}.
\]
This leads to
\begin{align}
\label{pt5}\frac{g(\gamma^{\prime}(t),\frac{D\gamma^{\prime}(t)}{dt}%
)}{g(\gamma^{\prime}(t),\gamma^{\prime}(t))}=\frac{1}{2}\big\{\frac{d}{dt}%
\ln(-g(\gamma^{\prime}(t),\gamma^{\prime}(t)))-\omega(\gamma^{\prime
}(t))\big\}.
\end{align}
Going back to (\ref{pt4}) and using (\ref{pt5}) we get
\begin{align}
\label{pt6}\frac{d^{2}\mu}{dt^{2}}-\frac{1}{2}\big\{\frac{d}{dt}\ln
(-g(\gamma^{\prime}(t),\gamma^{\prime}(t)))-\omega(\gamma^{\prime
}(t))\big\}\frac{d\mu}{dt}=0.
\end{align}
In order to integrate this equation, first define $\psi\doteq\frac{d\mu}{dt}$,
and then (\ref{pt6}) is reduced to a first order linear ordinary differential
equation for $\psi$, which can easily be integrated, yielding
\begin{align}
\label{ptdiffeq1}\frac{d\mu(t)}{dt}=\frac{d\mu(t_{0})}{dt}\Big[\frac
{g(\gamma^{\prime}(t),\gamma^{\prime}(t))}{g(\gamma^{\prime}(t_{0}%
),\gamma^{\prime}(t_{0}))}\Big]^{\frac{1}{2}}e^{-\frac{1}{2}\int_{t_{0}}%
^{t}\omega(\gamma^{\prime}(u))du}.
\end{align}
We can now integrate this equation once more and get the general solution for
(\ref{pt4}). It is not difficult to see that doing this we get the following:
\[
\mu(t)=\frac{\frac{d\mu(t_{0})}{dt}}{(-g(\gamma^{\prime}(t_{0}),\gamma
^{\prime}(t_{0})))^{\frac{1}{2}}}\int_{t_{0}}^{t}e^{-\frac{1}{2}\int_{t_{0}%
}^{u}\omega(\gamma^{\prime}(s))ds}(-g(\gamma^{\prime}(u),\gamma^{\prime}(u)))
^{\frac{1}{2}}du+\mu_{0},
\]
If $\Delta\tau(t)$ denotes the elapsed proper time between $t_{0}$ and $t,$ we
can write
\begin{equation}
\Delta\tau(t)=\frac{\frac{d\tau(t_{0})}{dt}}{(-g(\gamma^{\prime}(t_{0}%
),\gamma^{\prime}(t_{0})))^{\frac{1}{2}}}\int_{t_{0}}^{t}e^{-\frac{1}{2}%
\int_{t_{0}}^{u}\omega(\gamma^{\prime}(s))ds}(-g(\gamma^{\prime}%
(u),\gamma^{\prime}(u)))^{\frac{1}{2}}du. \label{explicitpt}%
\end{equation}

As a final comment about the mathematical consistency of our present
definition of proper time, we remark that the expression (\ref{explicitpt}),
which is the expression to be used when computing the elapsed proper time
measured by an observer between two events, is invariant both under the Weyl
transformations (\ref{weyl group})\ and reparametrizations of $\gamma$.
Therefore, we can say that we have a mathematically consistent definition of
proper time in the framework of Weyl geometry. Let us now analyze whether it
is physically sensible to adopt this definition. In order to address this
question we will consider three points.

\subsection*{i) The Riemannian limit}

We claim that the above definition of proper time coincides with the usual
definition adopted in general relativity and other metric theories of gravity
where the underlying space-time structure is that of a Riemannian manifold. In
order to see this, consider this definition when the manifold $(M,g)$ is
Riemannian. Let $\gamma$ be a time-like curve. In this case, the compatibility
of $\nabla$ with $g$ gives
\begin{equation}
g(\frac{D\gamma^{\prime}}{dt}(t),\gamma^{\prime}(t))=\frac{1}{2}\frac{d}%
{dt}g(\gamma^{\prime}(t),\gamma^{\prime}(t)).
\end{equation}
So $\gamma(t)$ is standard clock if and only if $g(\gamma^{\prime}%
(t),\gamma^{\prime}(t))$ is constant along $\gamma(t)$. This means that
$\gamma$ is parametrized by arc-length, or an affine reparametrization of it.
We conclude that our definition is consistent with the Riemannian limit.

\subsection*{ii)\ The WIST limit}

If the 1-form field $\omega$ is exact, i.e, if $\omega=d\phi$, where $\phi$ is
some smooth scalar field $\phi$ on $M$, then we say that the Weyl structure is
\textit{integrable,} and accordingly the resulting space-time is called
\textit{Weyl Integrable Space-Time (WIST)}. This kind of geometry has recently
attracted the attention of some cosmologists (see, for instance, \cite{Salim-Poulis},\cite{CBPF},\cite{Scholz1},\cite{Romero1},\cite{Romero2}). A particularly interesting recent review on this topic can be found in \cite{Scholz2}. In the case of WIST, it is already known that it is possible
to define the proper time interval between two events along a curve
$\gamma(t)$ in an invariant way as \cite{Romero2}
\begin{equation}
\Delta\tau=\int_{u_{1}}^{u_{2}}e^{-\frac{1}{2}\phi(\gamma(u))}\sqrt
{-g(\gamma^{\prime}(u),\gamma^{\prime}(u))}du.
\end{equation}
It is easy to see that this quantity is invariant under the group of Weyl
transformations defined in (\ref{weyl group}). We want to show that Perlick's
definition reduces to this expression in the case of a WIST model. In order to
do this, consider the explicit expression (\ref{explicitpt}) and set
$\omega=d\phi$. This will lead to%

\[%
\begin{split}
\Delta\tau &  =\frac{\frac{d\tau(u_{0})}{du}}{(-g(\gamma^{\prime}%
(u_{0}),\gamma^{\prime}(u_{0}))e^{-\phi(\gamma(u_{1}))})^{\frac{1}{2}}}%
\int_{u_{1}}^{u_{2}}e^{-\frac{1}{2}\phi(\gamma(u^{\prime}))}(-g(\gamma
^{\prime}(u^{\prime}),\gamma^{\prime}(u^{\prime}))^{\frac{1}{2}}du^{\prime}\\
&  =C\int_{u_{1}}^{u_{2}}e^{-\frac{1}{2}\phi(\gamma(u^{\prime}))}%
(-g(\gamma^{\prime}(u^{\prime}),\gamma^{\prime}(u^{\prime}))^{\frac{1}
{2}}du^{\prime}%
\end{split}
\]
which, by an appropriate setting of the scale to make $C=1,$ reproduces the
WIST definition of proper time.

\subsection*{iii) Additivity}

It seems that in any plausible physical definition proper time intervals
should be additive. This means that if an observer experiences three events
$A,B$ and $C$ in that order, then given two identical clocks, the time
interval measured by a single clock from $A$ to $C$ should be the same as the
sum of the time intervals measured by the other from $A$ to $B$, and from $B$
to $C$. Thus, for the above definition of proper time to be acceptable, it
must be the case that if we use (\ref{explicitpt}) to make computations in
both situations, the results should be the same.

\begin{figure}[th]
\centering
\includegraphics[width=40mm]{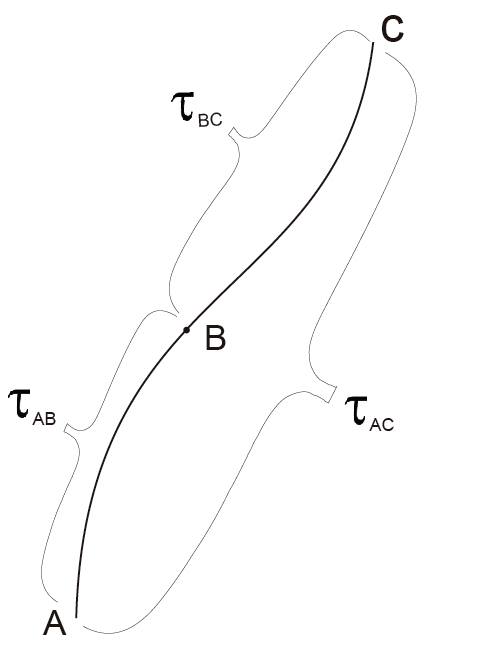}\caption{Additivity of proper
time.}%
\label{overflow}%
\end{figure}

In order to check that the definition that we have presented for proper time
is additive, we need to compute the proper time elapsed between $A$ and $B$,
and between $B$ and $C$, separately by using (\ref{explicitpt}), which would
represent the readings of two standard clocks used to compute these elapsed
times. Then, we have to add these results, and this should equal the result we
would get if we were to make a single computation of  the proper time elapsed
between $A$ and $C$, which represents the reading of a single clock. Carrying
out these calculations we get%

\[%
\begin{split}
\tau_{AB}  &  =\frac{\tau^{\prime}(t_{A})}{(-g(\gamma^{\prime}(t_{A}%
),\gamma^{\prime}(t_{A}))^{\frac{1}{2}}}\int_{t_{A}}^{t_{B}}e^{-\frac{1}%
{2}\int_{t_{A}}^{t}\omega(\gamma^{\prime}(s))ds}(-g(\gamma^{\prime}%
(t),\gamma^{\prime}(t)))^{\frac{1}{2}}dt\\
\tau_{BC}  &  =\frac{\tau^{\prime}(t_{B})}{(-g(\gamma^{\prime}(t_{B}%
),\gamma^{\prime}(t_{B}))^{\frac{1}{2}}}\int_{t_{B}}^{t_{C}}e^{-\frac{1}%
{2}\int_{t_{B}}^{t}\omega(\gamma^{\prime}(s))ds}(-g(\gamma^{\prime}%
(t),\gamma^{\prime}(t)))^{\frac{1}{2}}dt.\\
&
\end{split}
\]
On the other hand, from (\ref{ptdiffeq1}) we know that
\[
\tau^{\prime}(t_{B})=\tau^{\prime}(t_{A})\Big[\frac{g(\gamma^{\prime}%
(t_{B}),\gamma^{\prime}(t_{B}))}{g(\gamma^{\prime}(t_{A}),\gamma^{\prime
}(t_{A}))}\Big]^{\frac{1}{2}}e^{-\frac{1}{2}\int_{t_{A}}^{t_{B}}\omega
(\gamma^{\prime}(s))ds}.
\]
After substituting the above equation into the expression for $\tau_{BC}$ we
have that
\[%
\begin{split}
\tau_{AB}+\tau_{BC}  &  =\frac{\tau^{\prime}(t_{A})}{(-g(\gamma^{\prime}%
(t_{A}),\gamma^{\prime}(t_{A}))^{\frac{1}{2}}}\int_{t_{A}}^{t_{C}}e^{-\frac
{1}{2}\int_{t_{A}}^{t}\omega(\gamma^{\prime}(s))ds}(-g(\gamma^{\prime
}(t),\gamma^{\prime}(t)))^{\frac{1}{2}}dt\\
&  =\tau_{AC}%
\end{split}
\]
which proves the additivity of proper time intervals.

From the above arguments we are led to conclude that Perlick's definition of
proper time is both mathematically consistent and physically sensible. In the
next section, we will show that this definition is, in fact, equivalent to the
definition proposed by Ehlers, Pirani and Schild in their original paper
\cite{EPS}.

\section{Equivalence of Perlick's and EPS proper time.}

Let us begin by recalling the definition of proper time given by EPS
\cite{EPS}: We say that a time-like curve $\gamma$ is parametrized by proper
time if the tangent vector $\gamma^{\prime}$ is \textit{congruent} at each
point of the curve to a non-null vector $V$ which is parallel-transported
along $\gamma$. In the previous sentence, \textit{congruence} of vectors at a
point means that both vectors have the same \textit{norm}, \textit{i.e},
$g_{p}(\gamma_{p}^{\prime},\gamma_{p}^{\prime})=g_{p}(V_{p},V_{p})$. We will
now establish the equivalence between both definitions.

\begin{prop}
A time-like curve is parametrized by proper time according to EPS if it is
parametrized by proper time acordding to definition 2.
\end{prop}

\begin{proof}
Suppose that we have a time-like curve $\gamma$ parametrized by proper time according to EPS and let $\tau$ denote such parametrization. Then, by definition, there exists a parallel vector field $V$ along $\gamma$ satisfying the following condition:
\begin{align*}
g(\gamma'(\tau),\gamma'(\tau))=g(V(\tau),V(\tau)).
\end{align*}
Thus, differentiating the previous identity, using Weyl's compatibility condition and the fact that $V$ is a parallel vector field, we get the following.
\begin{align*}
\begin{split}
\omega(\gamma'(\tau))g(V(\tau),V(\tau))&=2g(\gamma'(\tau),\frac{D\gamma'(\tau)}{d\tau})+\omega(\gamma'(\tau))g(\gamma'(\tau),\gamma'(\tau))\\
&=2g(\gamma'(\tau),\frac{D\gamma'(\tau)}{d\tau})+\omega(\gamma'(\tau))g(V(\tau),V(\tau))\Rightarrow\\
0&=g(\gamma'(\tau),\frac{D\gamma'(\tau)}{d\tau})
\end{split}
\end{align*}
Hence $\gamma$ is parametrized by proper time according to definition 2.
In order to prove the converse, suppose that $\gamma$ is parametrized by proper time according to definition 2. Then, by hypothesis, $\gamma'(\tau)$ and $\frac{D\gamma'(\tau)}{d\tau}$ are orthogonal. Now consider the following initial value problem:
\begin{align*}
\frac{DV(\tau)}{d\tau}&=0,\\
V(\tau_0)&=\gamma'(\tau_0),
\end{align*}
which defines a unique parallel vector field along $\gamma$. Since $V$ is a parallel vector, using Weyl's compatibility condition, we see that $g(V(\tau),V(\tau))$ satisfies the following equation
\begin{align*}
\frac{d}{d\tau}g(V(\tau),V(\tau))=\omega(\gamma'(\tau))g(V(\tau),V(\tau)).
\end{align*}
Using the compatibility condition and the fact that $\gamma'(\tau)$ and $\frac{D\gamma'(\tau)}{d\tau}$ are orthogonal, we see that $g(\gamma'(\tau),\gamma'(\tau))$ also satisfies the previous equation. Furthermore, both solutions initially agree, that is, $g(V(\tau_0),V(\tau_0))=g(\gamma'(\tau_0),\gamma'(\tau_0))$. Then, by uniqueness of solutions, we get that $g(V(\tau),V(\tau))=g(\gamma'(\tau),\gamma'(\tau))$, which means that $\gamma$ is parametrized by proper time according to EPS.
\end{proof}

With the previous result we have established the equivalence of both definitions. This is an interesting fact, since, at first sight, the two definitions do not seem to be intimately related to each other. It is also worth noting that these two definitions have been widely used in the literature (see, for instance, \cite{EPS},\cite{Perlick1},\cite{F-I},\cite{Arg},\cite{Salim-Poulis}). However, it seems to us, that when an author accepts one of them, no reference to the other is made. This, in principle, could present an ambiguity if both definitions were not equivalent, since we would have different ways of defining proper time in this context. Thus, the establishing of the equivalence between the two definitions takes away any possible ambiguity.

\section{Analysis of the Second Clock Effect}

In this section, we investigate the question of whether or not a space-time
modelled as a Weyl structure, in which proper time is understood according to
Perlick's definition, exhibits the so-called \textit{second clock effect}. As
is well known, we say that a space-time model exhibits the second clock effect
if the clock rate of clocks depends on their histories \cite{Penrose}. The
analysis of whether a Weyl structure presents such an effect has been usually
done in a more or less \textit{intuitive} way, using the fact that the ``norms"
of parallel transported vectors depend on the path along which the transport
is taken, and accepting that such a parallel transported vector along a
time-like curve represents the clock rate of a clock. Thus, in this case, the
spectral lines emitted by an atomic clock could depend on its history, which
is something that has not been observed. Even though this argument might seem compelling, it is not a rigorous proof, and depends on some untested hypotheses, such as the fact that the tick
tack of a clock can be related with the norm of a parallel transported vector
field along the worldline of the clock, let alone the fact that such a
discussion should be made using concepts which are well-defined within the
context of Weyl's geometry. In this section, we intend to make a detailed
analysis of the second clock effect, and settle the question of whether or not
a general Weyl structure leads to the second-clock effect.

In order to give an answer to this question consider the following situation.
Suppose that we transport two identical standard clocks along a time-like
curve from $A$ to $B$, and then, at $B$, they separate, following different
paths $\gamma_{1}$ and $\gamma_{2}$ until they merge again at event $C$ (see
Fig. \ref{secondclock}), after which they continue their journey together
along the same path.\ Suppose that both clocks were synchronized at $A$.
Thinking of a clock as a device that counts the number of cycles of some
periodic process, what we are saying when we refer to synchronization is that
identical clocks use the same type of process and that the periods of these
cycles were set to be equal at $A$ (both clocks are set with the same scale at
$A$). Now assume that our space-time model does not exhibit a second clock
effect. Accepting this hypothesis means that we would expect that the clock
rate of a clock, at a given event in space-time, should depend only on
\textit{local properties} of the clock, that is, its position, instant
velocity, instant acceleration, etc, but not on its history. Therefore, in the
particular case we are considering, after we bring back the two identical
clocks together at $C,$ and keep them together, we would expect their clock
rates to coincide. In other words, we would expect the number of cycles
counted by either clock after $C$ to be the same. This also means that the
readings of the two clocks would coincide at any subsequent event $D$
($\tau_{CD}=\overline{\tau}_{CD}$). These considerations give us a way to test
a possible existence of the second clock effect: compute the elapsed time for
both clocks between $C$ and some subsequent event $D$ and see whether they
agree or not. If they do not, then clearly there is a second clock effect.

\begin{figure}[th]
\centering
\includegraphics[width=40mm]{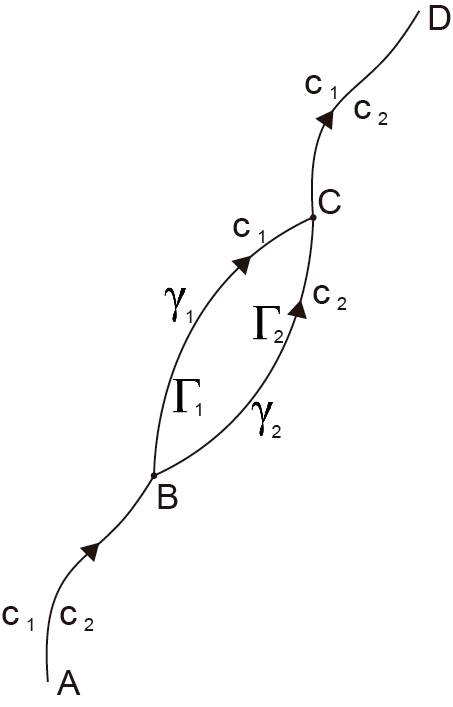}\caption{Transport of two clocks
c1 and c2. The curves $\gamma_{1}$ and $\gamma_{2}$ represent the world lines
of c1 and c2 respectively, while $\Gamma_{1}$ and $\Gamma_{2}$ denote the
portions of these curves between $B$ and $C$. }%
\label{secondclock}%
\end{figure}

Let us now, with the help of (\ref{explicitpt}), carry out some calculations
to make the above discussion more precise. First, let us recall that after the
event $C$ both clocks are being transported along the same time-like curve
$\gamma$. Suppose that the curve $\gamma$ is parametrized by some arbitrary
parameter $u$. Then, from (\ref{explicitpt}) we have%

\[%
\begin{split}
\tau &  =\frac{\tau^{\prime}(u_{C})}{(-g(\gamma^{\prime}(u_{C}),\gamma
^{\prime}(u_{C})))^{\frac{1}{2}}}\int_{u_{C}}^{u}e^{-\frac{1}{2}\int_{u_{c}%
}^{u^{\prime}}\omega(\gamma^{\prime}(s))ds}(-g(\gamma^{\prime}(u^{\prime
}),\gamma^{\prime}(u^{\prime})))^{\frac{1}{2}}du^{\prime}\\
\overline{\tau}  &  =\frac{\overline{\tau}^{\prime}(u_{C})}{(-g(\gamma
^{\prime}(u_{C}),\gamma^{\prime}(u_{C})))^{\frac{1}{2}}}\int_{u_{C}}%
^{u}e^{-\frac{1}{2}\int_{u_{c}}^{u^{\prime}}\omega(\gamma^{\prime}%
(s))ds}(-g(\gamma^{\prime}(u^{\prime}),\gamma^{\prime}(u^{\prime})))^{\frac
{1}{2}}du^{\prime}%
\end{split}
\]
where $\tau$ and $\overline{\tau}$ represent the readings of clocks 1 and 2
respectively. We can also compute $\tau_{AC}$ and $\overline{\tau}_{AC}$ using
(\ref{explicitpt}). In order to make this computation, we will consider that
the parametrization $u$ used for $\gamma$ after $C$ is the parametrization
used to parametrize the whole time-like path of clock 1. On the other hand,
for the path of clock 2 we will use $\overline{u}$ as a parameter. Then, a
straightforward calculation shows that
\[%
\begin{split}
\tau^{\prime}(u_{C})  &  =\tau^{\prime}(u_{A})\Big(\frac{g(\gamma_{1}^{\prime
}(u_{C}),\gamma_{1}^{\prime}(u_{C}))}{g(\gamma_{1}^{\prime}(u_{A}),\gamma
_{1}^{\prime}(u_{A}))}\Big)^{\frac{1}{2}}e^{-\frac{1}{2}\int_{u_{A}}^{u_{C}%
}\omega(\gamma_{1}^{\prime}(s))ds}\\
\overline{\tau}^{\prime}(\overline{u}_{C})  &  =\overline{\tau}^{\prime
}(\overline{u}_{A})\Big(\frac{g(\gamma_{2}^{\prime}(\overline{u}_{C}%
),\gamma_{2}^{\prime}(\overline{u}_{C}))}{g(\gamma_{2}^{\prime}(\overline
{u}_{A}),\gamma_{2}^{\prime}(\overline{u}_{A}))}\Big)^{\frac{1}{2}}%
e^{-\frac{1}{2}\int_{\overline{u}_{A}}^{\overline{u}_{C}}\omega(\gamma
_{2}^{\prime}(s))ds}%
\end{split}
\]
where $\gamma_{1}$ and $\gamma_{2}$ are the curves representing the world
lines of each clock. Now, since from $A$ to $B$ and after $C$ both world lines
are the same, both curves being equal in these intervals, we could
reparametrize $\gamma_{2}$ by changing from $\overline{u}$ to $u$ to obtain
\[
\overline{\tau}^{\prime}(\overline{u}_{C})=\overline{\tau}^{\prime}%
(u_{A})\frac{du(\overline{u}_{C})}{d\overline{u}}\Big(\frac{g(\gamma
_{2}^{\prime}(u_{C}),\gamma_{2}^{\prime}(u_{C}))}{g(\gamma_{2}^{\prime}%
(u_{A}),\gamma_{2}^{\prime}(u_{A}))}\Big)^{\frac{1}{2}}e^{-\frac{1}{2}%
\int_{\overline{u}_{A}}^{\overline{u}_{C}}\omega(\gamma_{2}^{\prime}(s))ds}.
\]
Thus we have
\[
\overline{\tau}^{\prime}(u_{C})=\frac{d\overline{u}(u_{C})}{du}\overline{\tau
}^{\prime}(\overline{u}_{C})=\overline{\tau}^{\prime}(u_{A})\Big(\frac
{g(\gamma_{2}^{\prime}(u_{C}),\gamma_{2}^{\prime}(u_{C}))}{g(\gamma
_{2}^{\prime}(u_{A}),\gamma_{2}^{\prime}(u_{A}))}\Big)^{\frac{1}{2}}%
e^{-\frac{1}{2}\int_{\overline{u}_{A}}^{\overline{u}_{C}}\omega(\gamma
_{2}^{\prime}(s))ds}.
\]
Since $\gamma_{1}^{\prime}(u_{A})=\gamma_{2}^{\prime}(u_{A})$ and $\gamma
_{1}^{\prime}(u_{C})=\gamma_{2}^{\prime}(u_{C})=\gamma^{\prime}(u_{C})$
(recalling that it is the same curve with the same parametrization), then for
the reading of clock 2 we get
\[%
\begin{split}
\overline{\tau} &  =\overline{\tau}^{\prime}(u_{A})\Big(\frac{g(\gamma
_{1}^{\prime}(u_{C}),\gamma_{1}^{\prime}(u_{C}))}{g(\gamma_{1}^{\prime}%
(u_{A}),\gamma_{1}^{\prime}(u_{A}))}\Big)^{\frac{1}{2}}e^{-\frac{1}{2}%
\int_{\overline{u}_{A}}^{\overline{u}_{C}}\omega(\gamma_{2}^{\prime}%
(s))ds}\frac{\int_{u_{C}}^{u}e^{-\frac{1}{2}\int_{u_{c}}^{u^{\prime}}%
\omega(\gamma^{\prime}(s))ds}(-g(\gamma^{\prime}(u^{\prime}),\gamma^{\prime
}(u^{\prime})))^\frac{1}{2}du^{\prime}}{(-g(\gamma^{\prime}(u_{C}),\gamma^{\prime
}(u_{C})))^{\frac{1}{2}}}\\
&  =\frac{\overline{\tau}^{\prime}(u_{A})}{\tau^{\prime}(u_{A})}e^{\frac{1}%
{2}\int_{u_{A}}^{u_{C}}\omega(\gamma_{1}^{\prime}(s))ds-\frac{1}{2}%
\int_{\overline{u}_{A}}^{\overline{u}_{C}}\omega(\gamma_{2}^{\prime}%
(s))ds}\tau.
\end{split}
\]
Also, as both clocks have the same scale at the event $A$, that is,
\[
\frac{\overline{\tau}^{\prime}(u_{A})}{\tau^{\prime}(u_{A})}=\frac
{d\overline{\tau}(\tau_{A})}{d\tau}=1\;,
\]
we finally get
\begin{equation}
\overline{\tau}=e^{\frac{1}{2}\int_{\Gamma_{1}}\omega-\frac{1}{2}\int
_{\Gamma_{2}}\omega}\tau. \label{secondclockeffect}%
\end{equation}
Thus we are led to conclude that

\begin{center}
\textit{A Weyl space-time does not exhibit a second clock effect if, and only
if, $\int_{\Gamma_{1}}\omega=\int_{\Gamma_{2}}\omega,$ $where$ $\Gamma_{1}$
and $\Gamma_{2}$ are arbitrary time-like curves joining the same pair of
events $A$ and $B$.}
\end{center}

Taking into consideration the previous statement, we are now in position to
discuss whether or not a Weyl structure is a suitable model for space-time. In
this framework we have an unambiguous, well-defined and physically sensible
definition for proper time, and we have arrived at the conclusion that a
general Weyl space-time exhibits a second clock effect. We note that up to now
this kind of effect has never been measured, and we could also give
interesting arguments, which would, at least, set strong constraints on the
values of this effect (see, for example, \cite{Geroch}). Then, it is natural
to ask, what would be the most general mathematical structure of space-time if
we were to reject space-times exhibiting a second clock effect. If, in the
previous statement, we could drop the condition that the curves $\Gamma_{1}$
and $\Gamma_{2}$ are time-like, then it is immediate to see that the
non-existence of a second clock effect would imply that the 1-form $\omega$
must be be closed. If, in turn, we assume that space-time is simply connected,
then we would be led to a Weyl integrable space-time. In what follows, we will
show that even if the timelike character of the curves is kept, a Weyl
integrable space-time still emerges as a consequence of the non-existence of a
second clock effect. Thus, let us suppose that $\omega$ is path independent
when integrated over time-like curves and see what this implies. First
consider an arbitrary event in space-time, $p\in M$, and the set of events
$I_{p}$\ that are causally connected with $p$. Thus $I_{p}$ is defined as
\[
I_{p}\doteq I_{p}^{+}\bigcup I_{p}^{-}=\{q\in
M,\;\text{such\ that\ there\ is\ a\ time-like\ curve\ joining\ }%
q\text{\ and\ }p\text{\},}%
\]
which is an open subset of $M$. On this set define the following function:
\[
f(q)\doteq\int_{\gamma}\omega
\]
where $\gamma$ is any time-like curve joining $p$ and $q$ (see figure
\ref{definition}). Since, by hypothesis, the integral of $\omega$ over any such
time-like curve does not depend on the choice of the curve, then the previous
function is well-defined.

\begin{figure}[th]
\centering
\includegraphics[width=45mm]{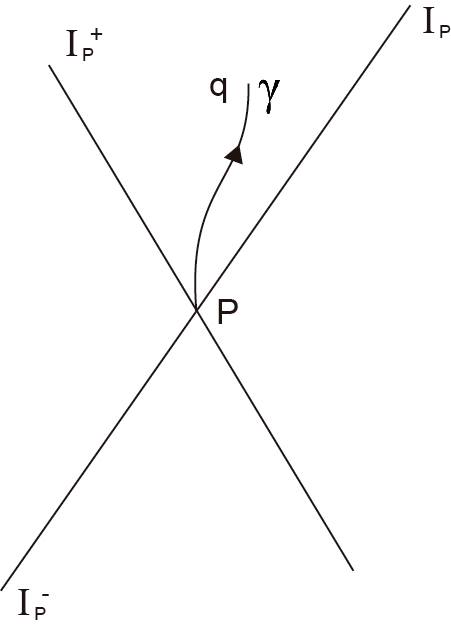}\caption{Domain of $f$. }%
\label{definition}%
\end{figure}

Before going any further, it is important to remark that $f$ is differentiable
near $p$. To see this, we explicitly compute $f$ in a normal coordinate system
around $p$. In such a coordinate system computing $f$ is not difficult since
we can choose as the time-like curve $\gamma$ joining $p$ and $q$, with $q$
sufficiently near to $p$, the unique time-like geodesic joining these points
inside the normal neighborhood. Then, since the coordinate expression of a
geodesic in a normal coordinate system is given by a straight line in the
direction of the initial velocity of the geodesic (see, for example,
\cite{Oneill}), we can make an easy explicit computation for $f$. Doing this,
we get the following:%

\begin{align}
\label{coordexpression}f(x)  &  =x^{\alpha}\int_{0}^{1}\omega_{\alpha}(tx)dt.
\end{align}

Assuming $\omega$ to be smooth, it is clear that $f$ is differentiable near
$p$. From now on we will restrict $f$ to such a neighborhood of $p$, so that
we know that $df$ exists. Consider now a neighborhood $U_{q}$ of $q$ such that
$U_{q}\subset I_{p}$ and a time-like vector $X_{q}\in T_{q}M$. Then, we can
construct a smooth curve $\mu$ defined in a neighborhood $J_{q}$ of the origin
satisfying
\[%
\begin{split}
\mu(0)  &  =q\\
\mu^{\prime}(0)  &  =X_{q}%
\end{split}
\]
Also, since $g_{q}(X(q),X(q))<0$, then, by continuity, there is a neighborhood
of $0\in\mathbb{R}$ where $\mu$ is time-like. To avoid complications in the
notation we will take $J_{q}$ to be such a neighborhood (if needed, we could
shrink the original $J_{q}$ in order for this to be true). Now we wish to compute%

\[
df_{q}(X_{q})=\frac{df\circ\mu(t)}{dt}|_{t=0}.
\]

We can compute $f\circ\mu(t)$ in the following way. First consider a piecewise
smooth time-like curve joining $p$ and $\mu(t)$ consisting of a time-like
curve $\beta$, joining $p$ and $q$, and then the curve $\mu$ joining $q$ and
$\mu(t)$. In this set-up we have the following:
\[%
\begin{split}
f\circ\mu(t)  &  =\int_{\beta}\omega+\int_{\mu}\omega\\
&  =f(q)+\int_{o}^{t}\omega(\mu^{\prime}(u))du.
\end{split}
\]
Therefore
\begin{align*}
\frac{df\circ\mu(t)}{dt}|_{t=0}  &  =\omega_{\mu(0)}(\mu^{\prime}(0))\\
&  =\omega_{q}(X_{q}).
\end{align*}
Thus we have shown that $df_{q}(X_{q})=\omega_{q}(X_{q})$, for any time-like
$X_{q}\in T_{q}M.$

To see how $df_{q}$ acts on an arbitrary vector of $T_{q}M$ consider an
orthonormal basis $\{e_{0},e_{i}\}$ of $T_{q}M$, where $e_{0}$ is time-like
and $e_{i}$ are space-like. Now if we consider the set of vectors in $T_{q}M$
given by
\[%
\begin{split}
\tilde{e}_{0}  &  \doteq e_{o}\\
\tilde{e}_{i}  &  \doteq2e_{0}+e_{i}%
\end{split}
\]
then $\{\tilde{e}_{o},\tilde{e}_{i}\}$ gives a basis of time-like vectors for
$T_{q}M$. Then if we pick $V\in T_{q}M$ arbitrary and we write it in this
basis $V=V^{\alpha}\tilde{e}_{\alpha}$, we can then compute the action of
$df_{q}$ on this element:
\[%
\begin{split}
df_{q}(V)  &  =V^{\alpha}df_{q}(\tilde{e}_{\alpha})\\
&  =V^{\alpha}\omega_{q}(\tilde{e}_{\alpha})\\
&  =\omega_{q}(V)\Rightarrow\\
df_{q}(V)  &  =\omega_{q}(V)\;\forall\;V\in\;T_{q}M\;and\;q\in\;I_{p}.
\end{split}
\]
Therefore we have shown that, given any point $q\in I_{p}$, there is a
neighborhood of $q$ where $\omega$ is exact. Then, if we accept that any event
$q$ in space-time lies in $I_{p}$ for some other event $p$, we conclude that
$\omega$ is closed. Finally, assuming that space-time is simply connected,
then, to avoid the second clock effect, our space-time model has to be reduced
to a Weyl integrable structure $(M,[(g,\phi)])$.

\section{Final Comments}

In this paper we have revisited the notion of proper time in the framework of
Weyl's space-time within the axiomatic approach put forward by Elhers, Pirani
and Schild. We have shown that the EPS original definition of proper time is
equivalent to the definition proposed by V. Perlick, even though the latter
seems to be more easily motivated in the context of EPS's axiomatic approach
to space-time. After showing that Perlick's notion leads to a well-defined and
physically sensible definition of proper time in a general Weyl space-time,
without invoking any additional axioms, we proved that this kind of space-time
exhibits a second clock effect. We then derived the condition a space-time
must obey in order that a second clock effect does not appear. We have shown
that in this case the geometric space-time structure should be that of a Weyl
integrable space-time. Our final conclusion is that, within the slightly
modified EPS axiomatic approach, taking into account Perlick's proper time,
Weyl integrable space-time appears naturally as the most general model for
space-time. However, the question of a possible existence of a second clock
effect, which would then widen this scenario leaving us with a general Weyl
space-time as a suitable mathematical model, is something that should be
settled by experiment.

\bigskip

\section*{Acknowledgements}

\noindent R. A and C. R. would like to thank CNPq and CLAF for financial support.

\end{document}